\documentclass[11pt]{article}
\usepackage{amsmath}
\usepackage{amssymb}
\usepackage{amsfonts}
\usepackage{latexsym}
\newcommand{\be}{\begin{equation}}
\newcommand{\en}{\begin{equation}}
\newcommand{\mc}{\mathcal}
\newcommand{\mb}{\mathbb}
\newcommand{\D}{\mc D}
\newcommand{\Hil}{\mc H}
\newcommand{\A}{\mc A}
\newcommand{\Ao}{{{\mc A}_0}}
\newcommand{\LL}{\mc L}

\newcommand{\LD}{{\LL}^\dagger (\D)}
\newcommand{\1}{1\!\!\!\!1}
\newcommand{\LDD}{{\LL} (\D,\D')}
\newcommand{\LDH}{{\LL}^\dagger (\D,\Hil)}

\newtheorem{defn}{Definition}[section]
\newtheorem{prop}[defn]{Proposition}
\newtheorem{thm}[defn]{Theorem}
\newtheorem{lemma}[defn]{Lemma}

\newenvironment{proof}{\noindent {\bf Proof --}}{\hfill$\square$ \vspace{3mm}\endtrivlist}

\begin{document}

\thispagestyle{empty}

\vspace*{0.7cm}

\begin{center}
{\Large \bf Derivations of quasi *-algebras}

\vspace{1.5cm}

{\large F. Bagarello }
\vspace{3mm}\\
 Dipartimento di Matematica e Applicazioni \\
Facolt\`a d' Ingegneria - Universit\`a di Palermo \\ Viale delle
Scienze,       \baselineskip15pt
     I-90128 - Palermo - Italy \vspace{2mm}\\

\vspace{5mm}

{\large A. Inoue}
\vspace{3mm}\\
Department of Applied Mathematics,\\ Fukuoka University,\\
J-814-80 Fukuoka, Japan

\vspace{3mm} and

\vspace{3mm} {\large C.Trapani}
\vspace{3mm}\\
 Dipartimento di Matematica e Applicazioni \\
Universit\`a di Palermo\\Via Archirafi 34, \baselineskip15pt
I-90123 - Palermo - Italy\vspace{2mm}

\end{center}

\vspace*{12mm}

 \noindent{\sc Abstract} \\\noindent The spatiality of
derivations of quasi *-algebras is investigated by means of
representation theory. Moreover, in view of physical applications,
the spatiality of the limit of a family of spatial derivations is
considered.

\vspace{3cm} \noindent{\em 2000 Mathematics Subject
Classification}: 47L60; 47L90.
 \vfill

\newpage

\section{Introduction}
In the so-called algebraic approach to quantum systems, one of the
basic problems to solve consists in the rigorous definition of the
algebraic dynamics, i.e. the time evolution of observables and/or
states. For instance, in quantum statistical mechanics or in
quantum field theory one tries to recover the dynamics by
performing a certain limit of the strictly {\em local} dynamics.
However, this can be successfully done only for few models and
under quite strong topological assumptions (see, for instance,
\cite{sakai} and references therein). In many physical models the
use of local observables corresponds, roughly speaking, to the
introduction  of some {\em cut-off} (and to its successive
removal) and this is in a sense a general and frequently used
procedure,
see \cite{thi,bagarello,bagatrap,lass1} for conservative and \cite{sew,sewbag} for dissipative systems.\\
Introducing a cut-off means that in the description of some
physical system, we know a {\em regularized} hamiltonian $H_L$,
where $L$ is a certain parameter closely depending on the nature
of the system under consideration. The role of the commutator
$[H_L,A]$, $A$ being an observable of the physical system
 (in a sense that will be made clearer in the following), is crucial in the analysis of the dynamics of the system.
 We have discussed several properties of this map in a recent paper, \cite{bagtra}, focusing our attention mainly on the
 existence of the algebraic dynamics $\alpha_t$ given a family of operators $H_L$ as above. Here, in a certain sense,
 we reverse the point of view. We start with a (generalized) derivation $\delta$ and we first consider the
 following problem: under which conditions is the map $\delta$ spatial (i.e., is implemented by a certain operator)?
 The spatiality of derivations is a very classical problem when
 formulated in *-algebras and it as been extensively studied in
 the literature in a large variety of situations, mostly depending
 on the topological structure of the *-algebras under
 consideration (C*-algebras, von Neumann algebras, O*-algebras,
 etc. See \cite{sakai,brarob,camillo, book}).
 In this paper we consider a more general set-up, turning our
 attention to derivations taking their values in a quasi
 *-algebra. This choice is motivated by possible applications to
 the physical situations described above. Indeed, if $\Ao$ denotes
 the *-algebra of local observables of the system, in order to
 perform the so-called thermodynamical limits of certain local
 observables, one endows $\Ao$ with a locally convex topology
 $\tau$, conveniently chosen for this aim (the so called {\em
 physical topology}). The completion $\A$ of $\Ao[\tau]$, where
 thermodynamical limits mostly live, may fail to be an algebra
 but it is in general a {\em quasi *-algebra} \cite{lass1, ctrev,
 book}.
 For these reasons we start with considering, given a quasi
 *-algebra $(\A,\Ao)$, a derivation $\delta$ defined in $\Ao$
 taking its values in $\A$, and investigate its spatiality. In
 particular, we consider the case where $\delta$ is the limit of a
 net $\{\delta_L\}$ of spatial derivations of $\Ao$ and give
 conditions for its spatiality and for the implementing operator
 to be the limit, in some sense, of the operators $H_L$ implementing
 the $\{\delta_L\}$'s.

The paper is organized as follows:

In the next section we give the essential definitions of the
algebraic structures needed in the sequel.

In Section 3, the possibility of extending $\delta$ beyond $\Ao$,
through a notion of $\tau-$closability is investigated.

Section 4 is devoted to the analysis of the spatiality of
*-derivations which are induced by *-representations, and of the
spatiality of the limit of a net of spatial *-derivations. We also
extend our results to the situation where the *-representation,
instead of living in Hilbert space, takes its values in a quasi
*-algebra of operators in rigged Hilbert space
(qu*-representation).

 \section{The mathematical framework}

Let $\A$ be a linear space and $\Ao$ a  $^\ast$ -algebra contained
in $\A$. We say that $\A$ is a quasi  $^\ast$ -algebra with
distinguished  $^\ast$ -algebra $\Ao$ (or, simply, over $\Ao$) if
\begin{itemize}\item[(i)] the left multiplication $ax$ and the right multiplication $
xa$ of an element $a$ of $\A$ and an element $x$ of $\A_0$ which
 extend the multiplication of $\A_0$ are always defined and
bilinear; \item[(ii)] $x_1 (x_2 a)= (x_1x_2 )a$ and $x_1(a
 x_2)= (x_1 a) x_2$, for each $x_1, x_2 \in \A_0$ and $a \in \A$;

\item[(iii)] an involution $*$ which extends the involution of $\A_0$
is defined in $\A$ with the property $(ax)^*= x^*a^*$ and $(xa)^
* =a^* x^*$ for each $x \in \A_0$ and $a \in \A$.
\end{itemize}
A quasi  $^\ast$ -algebra $(\A,\Ao)$ is said to have a unit
$\mathbb{I}$ if there exists an element $\mathbb{I} \in \Ao$ such
that $a\mathbb{I} =\mathbb{I} a=a, \;\, \forall a\in \A$. In this
paper we will always assume that the quasi $^\ast$ -algebra under
consideration have an identity.

Let $\A_0[\tau]$ be a locally convex $*$-algebra. Then the
completion $\overline{\A_0}[\tau]$ of $\A_0[\tau]$ is a quasi $*
$-algebra over $\A_0$ equipped with the following left and right
 multiplications: \\
for any $x \in \A_0$ and $a \in \A$,
$$
ax\equiv \lim_\alpha x_\alpha x \hspace{3mm}\mbox{ and
}\hspace{3mm} xa \equiv \lim_\alpha xx_\alpha,
$$
where $\{ x_\alpha \}$ is a net in $\A_0$ which converges to $a$
 w.r.t. the topology $\tau$.
Furthermore, the left and right multiplications are separately
continuous. A $*$-invariant subspace $\A$ of
$\overline{\A_0}[\tau]$ containing $\A_0$ is said to be a {\it
(quasi-) $*$-subalgebra} of $\overline{\A_0}[\tau]$ if $ax$ and
$xa$ in $\A$ for any $x \in \A_0$ and $a \in \A$. Then we have
\[
x_1(x_2a)= \lim_\alpha x_1 (x_2 x_\alpha)= \lim_\alpha(x_1x_2)x_
\alpha = (x_1x_2)a
\]
and similarly,
\begin{align*}
& (ax_1)x_2  = a (x_1 x_2), \\
& x_1 (ax_2)= (x_1 a)x_2
\end{align*}
for each $x_1, x_2 \in \A_0$ and $a \in \A$, which implies that
$\A$ is a quasi $*$-algebra over $\A_0$, and furthermore,
$\A[\tau]$ is a locally convex space containing $\A_0$ as dense
subspace and the right and left multiplications are separately
continuous. Hence, $\A$ is said to be a {\it locally convex quasi
$*$-algebra} over $\A_0$.

If $(\A[\tau], \Ao)$ is a locally convex quasi *-algebra, we
indicate with $\{p_\alpha, \alpha\in I\}$ a directed set of
seminorms which defines $\tau$.

In a series of papers (\cite{bt1}-\cite{bt4}) we have considered a
special class of quasi *-algebras, called CQ*- algebras, which
arise  as completions of C*-algebras. They can be introduced in
the following way:

Let $\A$ be a right Banach module over the C*-algebra $\A_\flat$
with involution $\flat$ and C*-norm $\|.\|_\flat$, and further
with isometric involution $*$ and such that $\A_\flat \subset \A$.
Set $\A_\sharp = (\A_\flat)^*$.  We say that $\{\A, *, \A_\flat,
\flat\}$ is a CQ*-algebra if
\begin{itemize}
\item[(i)] $\A_\flat$ is dense in $\A$ with respect to its norm $\|\,\|$,
\item[(ii)]$\A_o:=\A_\flat \cap \A_\sharp$ is dense in $\A_\flat$ with respect to its norm  $\|\,\|_\flat$,
\item[(iii)] $(ab)^* = b^*a^*, \quad \forall a,b \in \A_o$,
\item[(iv)]$\|y\|_\flat = \sup_{a \in \A, \|a\|\leq 1}\|ay\|,\quad y \in \A_\flat$.
\end{itemize}

Since $*$ is isometric, the space $\A_\sharp$ is itself, as it is easily seen, a C*-algebra with respect to the involution $x^\sharp:={(x^*)^\flat}^*$ and the norm $\|x\|_\sharp :=\|x^*\|_\flat$.\\
A CQ*-algebra is called {\em proper} if $\A_\sharp=\A_\flat$. When
also $\flat =\sharp$, we indicate a proper CQ*-algebra with the
notation $({\mc A},*,\Ao)$, since * is the only relevant
involution and $\Ao=\A_\sharp=\A_\flat$.

An example of CQ*-algebra is provided by certain subspaces of
${\mc B}({\mc H}_{+1},{\mc H}_{-1})$, ${\mc B}({\mc H}_{+1})$,
${\mc B}({\mc H}_{-1})$, the spaces  of operators acting on a
triplet (scale) of Hilbert spaces generated in canonical way by an
unbounded operator $S\geq \1$. For details,
see\cite{bt1,bt2,book}. From a purely algebraic point of view,
each CQ*-algebra can be considered as an example of partial
*-algebra, \cite{AK,book}, by which we mean a vector space $\A$
with involution $a \rightarrow a ^\ast$ [i.e. $(a+\lambda b) ^\ast
=a ^\ast +\overline{\lambda} b ^\ast$ ; $a=a^{\ast \ast}$ ] and a
subset  $\Gamma \subset\A\times\A$ such that (i)  $(a,b)\in
\Gamma$  implies $(b ^\ast ,a ^\ast )\in \Gamma$ ; (ii) $(a,b)$
and $(a,c)\in  \Gamma$ imply $(a,b+\lambda c)\in  \Gamma$ ; and
(iii) if $(a,b)\in \Gamma$, then there exists an element  $ab\in
\A$ and for this multiplication (which is not supposed to be
associative) the following properties hold:\\ if $(a,b)\in \Gamma$
and $(a,c)\in \Gamma$  then $ab+ac=a(b+c)$ and  $(ab) ^\ast =b
^\ast a ^\ast$ .

In the following we also need the concept of {\em
*-representation}.

Let $\D$ be a dense domain in Hilbert space $\Hil$. As usual we
denote with $\LD$ the space of all closable operators $A$ with
domain $\D$, such that $D(A^*)\supset \D$ and both $A$ and $A^*$
leave $\D$ invariant. As is known, $\D$ is a *-algebra with the
usual operations $A+B$, $\lambda A$, $AB$ and the involution
$A^\dagger=A^*|_\D$. Let now $\A$ be a locally convex quasi
*-algebra over $\Ao$ and $ \pi_o$ be a *-representation of $\Ao$,
that is, a *-homomorphism from $\Ao$ into the *-algebra $\LD$, for
some dense domain $\D$. In general, extending $ \pi_o $ beyond
$\Ao$ will force us to abandon the invariance of the domain $\D$.
That is, for $A \in \A \backslash \Ao$, the extended
representative $\pi(A)$ will only belong  to $\LDH$, which is
defined as the set of all closable operators $X$ in $\Hil$ such
that
 $D(X)  = \D$ and $D(X^*)  \supset  \D$ and it is a partial *-algebra (called partial O*-algebra on $\D$) with the usual operations
 $X+Y$, $\lambda X$, the involution $X^\dagger=X^*|\D$ and the weak product $X${\tiny$\Box$}$Y\equiv X^{\dagger *}Y$ whenever $Y\D\subset D(X^{\dagger *})$ and $X^\dagger\D\subset D(Y^{*})$.

It is also known that, defining on $\D$ a suitable (graph)
topology, one can build up the
 rigged Hilbert space $\D \subset \Hil  \subset \D'$, where  $\D'$ is the dual of $\D$, \cite{gelf},
and one has
$$
\LD \;  \subset \; \LDD,
$$
where $\LDD$ denotes the space of all continuous linear maps from
$\D$ into $\D'$. Moreover, under additional topological
assumptions, the following inclusions hold: $\LD \subset \; \LDH
\subset\;\LDD$. A more complete definition will be given in Section 4.

Let $(\A,\Ao)$ be a quasi *-algebra, $\D_\pi$ a dense domain in a
certain Hilbert space $\Hil_\pi$, and $\pi$ a linear map from $\A$
into $\LL^\dagger(\D_\pi, \Hil_\pi)$ such that:

(i) $\pi(a^*)=\pi(a)^\dagger, \quad \forall a\in \A$;

(ii) if $a\in \A$, $x\in \Ao$, then $\pi(a)${\tiny$\Box$}\!\!
$\pi(x)$ is well defined and $\pi(ax)=\pi(a)${\tiny$\Box$}\!\!
$\pi(x)$.

We say that such a map $\pi$ is a *-representation of $\A$.
Moreover, if

(iii) $\pi(\Ao)\subset \LL^\dagger(\D_\pi)$,

then $\pi$ is a *-representation of the quasi *-algebra
$(\A,\Ao)$.

Let $\pi$ be a *-representation of $\A$. The strong topology
$\tau_s$ on $\pi(\A)$ is the locally convex topology defined by
the following family of seminorms: $\{p_\xi(.); \;
\xi\in\D_\pi\}$, where $p_\xi(\pi(a))\equiv\|\pi(a)\xi\|$, where
$a\in \A$, $\xi\in \D_\pi$.

For an overview on partial *-algebras and related topics we refer
to \cite{book}.

\section{*-Derivations and their closability}

Let $(\A,\Ao)$ be a quasi *-algebra.

\begin{defn}
A {\em *-derivation of}  $\Ao$ is a map $\delta: \Ao\rightarrow
\A$
 with the following properties:
\begin{itemize}
\item[(i)]  $\delta(x^*)=\delta(x)^*, \; \forall x \in \Ao$;
\item[(ii)] $\delta(\alpha x+\beta y) = \alpha \delta( x)+\beta\delta( y), \; \forall x,y \in \Ao, \forall \alpha,\beta \in \mathbb{C}$;
\item [(iii)] $\delta(xy) = x\delta( y)+\delta( x)y,  \; \forall x,y \in \Ao$.
\end{itemize}
\label{Definition 3.1}
\end{defn}

As we see, the *-derivation is originally defined only on $\Ao$.
Nevertheless, it is clear that this is not the unique possibility
at hand: $\delta$ could also be defined on the whole $\A$, or in a
subset of $\A$ containing $\Ao$, under some continuity or
closability assumption. Since continuity of $\delta$ is a rather
strong requirement, we consider here a weaker condition:

\begin{defn}
A *-derivation $\delta$ of  $\Ao$ is said to be $\tau$-closable
if, for any net $\{x_\alpha\}\subset\Ao$ such that
$x_\alpha\stackrel{\tau}{\rightarrow} 0$ and
$\delta(x_\alpha)\stackrel{\tau}{\rightarrow} b\in\A$, one has
$b=0$. \label{Definition 3.2}
\end{defn}

If $\delta$ is a $\tau$-closable *-derivation then we define
\begin{equation}
D(\overline\delta)=\{a\in\A :\, \exists
\{x_\alpha\}\subset\Ao\mbox{ s.t. } \tau-\lim_\alpha x_\alpha=a
\mbox{ and } \delta(x_\alpha) \mbox{ converges in } \A\}.
\label{31}
\end{equation}
Now, for any $a\in D(\overline\delta)$, we put
\begin{equation}
\overline\delta(a)=\tau-\lim_\alpha \delta(x_\alpha), \label{32}
\end{equation}
and the following lemma holds:

\begin{lemma}

If $\delta(\Ao)\subset\Ao$ then $D(\overline\delta)$ is a quasi
*-algebra over $\Ao$.

\label{lemma31}
\end{lemma}

\begin{proof}
First we observe that $D(\overline\delta)$ is a complex vector
space. In particular, it is closed under involution. In fact, from
the definition itself, if $a\in D(\overline\delta)$ then there
exists a net $\{x_\alpha\}$ $\tau$-converging to $a$. But, since
the involution is $\tau$-continuous, the net $\{x_\alpha^*\}$ is
$\tau$-converging to $a^*\in\A$. We conclude that whenever $a$
belongs to $D(\overline\delta)$,  $a^*\in D(\overline\delta)$.

Next we show that the multiplication between an element $a\in
D(\overline\delta)$ and  $x\in\Ao$ is well-defined. We consider
here the product $ax$. The proof of the existence of $xa$ is
similar.

Since  $a\in D(\overline\delta)$ then  there exists
$\{x_\alpha\}\subset\Ao$ such that
$x_\alpha\stackrel{\tau}{\rightarrow} a$. Moreover the net
$\delta(x_\alpha)$ $\tau$-converges to an element $b\in\A$:
$\delta(x_\alpha)\stackrel{\tau}{\rightarrow}
b=\overline\delta(a)$. Recalling now that the multiplication is
separately continuous and since, by assumptions,
$\delta(x)\in\Ao$, we deduce that $\delta(x_\alpha
x)=\delta(x_\alpha)x+x_\alpha\delta(x)\stackrel{\tau}{\rightarrow}
\overline\delta(a)x+a\delta(x)$, which shows that $ax$ belongs to
$D(\overline\delta)$ and that
$\overline\delta(ax)=\tau-\lim_\alpha\delta(x_\alpha x)$.

\end{proof}

This Lemma shows that, under some assumptions, it is possible to
extend $\delta$ to a set larger than $\Ao$ which, also if it is
different from $\A$, is a  quasi *-algebra over $\Ao$ itself. This
result suggests the following rather general definition:

\begin{defn}
Let $(\A,\Ao)$ be a quasi *-algebra and $\D$ be a vector subspace
of $\A$ such that $(\D,\Ao)$ is a quasi *-algebra. A map $\delta:
\D\rightarrow \A$ is called a *-derivation if
\begin{itemize}
\item[(i)]  $\delta(\Ao)\subset\Ao$ and $\delta_0\equiv\delta|_\Ao$ is a *-derivation of $\Ao$;
\item[(ii)] $\delta$ is linear;
\item [(iii)] $\delta(ax) = a\delta( x)+\delta( a)x= a\delta_0( x)+\delta( a)x, \: \forall a\in \D$ and $\forall x\in \Ao$.
\end{itemize}
\label{Definition 3.3}
\end{defn}

\vspace{2mm}

{\bf Remark:--} Because of the previous results, if $\delta_0$ is
$\tau$-closable then its closure $\overline\delta_0$ is a
*-derivation defined on $D(\overline\delta_0)$.

\vspace{2mm}

Now we look for conditions for a *-derivation $\delta$ to be
closable, making use of some duality result. For that we first
recall that if $(\A[\tau],\Ao)$ is a locally convex quasi
*-algebra and $\delta$ is a *-derivation of $\Ao$, we can define
the adjoint derivation $\delta'$ acting on a subspace $D(\delta')$
of the dual space $\A'$ of $\A$.  The derivation $\delta'$ is
first defined, for $\omega\in\A'$ and $x\in\Ao$, by
$(\delta'\omega)(x)=\omega(\delta(x))$ and then extended to the
domain
$$
D(\delta')=\{\omega\in\A': \, \delta'\omega \mbox{ has a
continuous extension to } \A\}.
$$

A classical result, \cite{kothe}, states that $\delta$ is
$\tau$-closable if, and only if, $D(\delta')$ is
$\sigma(\A',\A)$-dense in $\A'$. We now prove the following
result.

\begin{prop}
Let $\delta:\Ao\rightarrow\A$ be a *-derivation. Assume that there
exists $\omega\in\A'$ such that $\omega|_{\Ao}$ is a positive
linear functional on $\Ao$ and

(1) $\omega\circ\delta$ is $\tau$-continuous on $\Ao$;

(2) the GNS-representation $\pi_\omega$ of $\A_0$ is faithful.

Then $\delta$ is $\tau$-closable.
\end{prop}

\begin{proof}
First we notice that condition (1) above implies that $\omega\in
D(\delta')$. Secondly, let $x,y,z\in\Ao$. Since
$\omega(x\delta(y)z)=\omega(\delta(xyz))-\omega(\delta(x)yz)-\omega(xy\delta(z))$,
we have, as a consequence of the continuity of $\omega\circ\delta$
and of $\omega$ itself:
$$
|\omega(x\delta(y)z)|\leq p_\alpha(xyz)+p_\beta(\delta(x)yz)+
p_\gamma(xy\delta(z))\leq C_{x,z} p_\sigma(y),
$$
where we  have also used the continuity of the multiplication.
$C_{x,z}$ is a suitable positive constant depending on both $x$
and $z$. Let us further define a new linear functional
$\omega_{x,z}(y)=\omega(xyz)$. Of course we have
$|\omega(xyz)|\leq D_{x,z}p_\alpha(y)$, for some seminorm
$p_\alpha$ and a positive constant $D_{x,z}$. It follows that
$\omega_{x,z}$ has a continuous extension to $\A$, which we still
denote with the same symbol. Moreover, since
$(\delta'\omega_{x,z})(y)=\omega_{x,z}(\delta(y))=\omega(x\delta(y)z)$,
we have $|(\delta'\omega_{x,z})(y)|\leq C_{x,z} p_\sigma(y)$, for
every $y\in\Ao$. This implies that $\omega_{x,z}$ belongs to
$D(\delta')$ or, in other words, that $\omega_{x,z}$ has a
continuous extension to $\A$. For this reason we have
$D(\delta')\supset$ linear span$\{\omega_{x,z}: \, x,z\in\Ao\}$,
and this set is dense in $\A'$. Were it not so, then there would
exists a non  zero element $y\in\Ao$ such that $\omega_{x,z}(y)=0$
for all $x,z\in\Ao$. But, this is in contrast with the
faithfulness of the GNS-representation $\pi_\omega$ since we would
also have
$\omega(xyz)=<\pi_\omega(y)\lambda_\omega(z),\lambda_\omega(x^*)>=0$
for all $x,z\in\Ao$, which, in turn, would imply that
$\pi_\omega(y)=0$.

\end{proof}

\section{Spatiality of *-derivations induced by *-representations}

Let $(\A,\Ao)$ be a quasi *-algebra and $\delta$ be a *-derivation
of $\Ao$ as defined in the previous section. Let $\pi$ be a
*-representation of  $(\A,\Ao)$. We will always assume that
whenever $x\in \Ao$ is such that $\pi(x)=0$,  $\pi(\delta(x))=0$
as well. Under this assumption, the linear map
\begin{equation}
\delta_\pi(\pi(x))=\pi(\delta(x)), \quad x\in \Ao, \label{41}
\end{equation}
is well-defined on $\pi(\Ao)$ with values in $\pi(\A)$ and it is a
*-derivation of $\pi(\Ao)$. We call $\delta_\pi$ the *-derivation
{\em induced} by $\pi$.

Given such a representation $\pi$ and its dense domain $\D_\pi$,
we consider the usual graph topology $t_\dagger$ generated by the
seminorms

\begin{equation}
\xi\in\D_\pi \rightarrow \|A\xi\|, \quad A\in \LL^\dagger(\D_\pi).
\label{42}
\end{equation}

Calling $\D_\pi'$ the conjugate dual of $\D_\pi$ we get the usual
rigged Hilbert space $\D_\pi[t_\dagger] \subset \Hil_\pi  \subset
\D_\pi'[t_\dagger']$, where $t_\dagger'$ denotes the strong dual
topology of $\D_\pi'$. As usual we denote with
$\LL(\D_\pi,\D_\pi')$ the space of all continuous linear maps from
$\D_\pi[t_\dagger]$ into $\D_\pi'[t_\dagger']$, and with
$\LL^\dagger(\D_\pi)$ the *-algebra of all  operators $A$ in
$\Hil_\pi$ such that both $A$ and its adjoint $A^*$ map $\D_\pi$
into itself. In this case, $\LL^\dagger(\D_\pi)\subset
\LL(\D_\pi,\D_\pi')$. Each operator $A\in \LL^\dagger(\D_\pi)$ can
be extended to all of $\D_\pi'$ in the following way:
$$
<\hat A\xi',\eta>=<\xi',A^\dagger \eta>, \quad \forall \xi'\in
\D_\pi', \quad \eta\in \D_\pi.
$$
Therefore the multiplication of  $X\in\LL(\D_\pi,\D_\pi')$ and
$A\in\LL^\dagger(\D_\pi)$ can always be defined:
$$
(X\circ A)\xi=X(A\xi), \mbox{ and } (A\circ X)\xi=\hat A(X\xi),
\quad \forall \xi\in \D_\pi.
$$

With these definitions it is known that
$(\LL(\D_\pi,\D_\pi'),\LL^\dagger(\D_\pi))$ is a quasi *-algebra.

We can now prove the following

\begin{thm}

Let $(\A,\Ao)$ be a locally convex quasi *-algebra with identity
and $\delta$ be a *-derivation of $\Ao$.

Then the following statements are equivalent:

(i) There exists a $(\tau-\tau_s)$-continuous, ultra-cyclic
*-representation $\pi$ of $\A$, with ultra-cyclic vector $\xi_0$,
such that the *-derivation $\delta_\pi$ induced by  $\pi$  is
spatial, i.e.

there exists $H=H^\dagger\in \LL(\D_\pi,\D_\pi')$ such that
$H\xi_0\in \Hil_\pi$ and

\begin{equation}
\delta_\pi(\pi(x))=i\{H\circ\pi(x)-\pi(x)\circ H\},\quad \forall
x\in\Ao. \label{43}
\end{equation}

\vspace{2mm}

(ii) There exists a positive linear functional $f$ on $\Ao$ such
that:

\begin{equation}
f(x^*x)\leq p(x)^2, \quad \forall x\in \Ao, \label{44}
\end{equation}
for some continuous seminorm $p$ of $\tau$ and, denoting with
$\tilde f$ the continuous extension of $f$ to  $\A$, the following
inequality holds:

\begin{equation}
|\tilde f(\delta(x))|\leq C(\sqrt{f(x^*x)}+\sqrt{f(xx^*)}), \quad
\forall x\in \Ao, \label{45}
\end{equation}
for some positive constant $C$.

\vspace{2mm}

(iii) There exists a positive sesquilinear form $\varphi$ on
$\A\times\A$ such that:

$\varphi$ is invariant, i.e.

\begin{equation}
\varphi(ax,y)=\varphi(x,a^*y), \mbox{ for all } a\in \A \mbox{ and
} x,y\in\Ao; \label{46}
\end{equation}

$\varphi$ is $\tau$-continuous, i.e.

\begin{equation}
|\varphi(a,b)|\leq p(a) p(b), \mbox{ for all } a,b\in \A,
\label{47}
\end{equation}
for some continuous seminorm $p$ of $\tau$; and $\varphi$
satisfies the following inequality:

\begin{equation}
|\varphi(\delta(x),\1)|\leq
C(\sqrt{\varphi(x,x)}+\sqrt{\varphi(x^*,x^*)}), \quad \forall x\in
\Ao, \label{48}
\end{equation}
for some positive constant $C$. \label{theorem41}
\end{thm}

\begin{proof} First we show that (i) implies (ii).

We recall that the ultra-cyclicity of the vector $\xi_0$ means
that $\D_\pi=\pi(\Ao)\xi_0$. Therefore, the map defined as

\begin{equation}
f(x)=<\pi(x)\xi_0,\xi_0>, \quad x\in \Ao, \label{49}
\end{equation}
is a positive linear functional on $\Ao$. Moreover, since
$f(x^*x)=\|\pi(x)\xi_0\|^2$, equation (\ref{44}) follows because
of the $(\tau-\tau_s)$-continuity of $\pi$. As for equation
(\ref{45}), it is clear first of all that $f$ has a unique
extension to $\A$ defined as
\begin{equation}
\tilde f(a)=<\pi(a)\xi_0,\xi_0>, \quad a\in \A, \label{410}
\end{equation}
due the $(\tau-\tau_s)$-continuity of $\pi$. Therefore we have,
using (\ref{43}),
\begin{eqnarray*}
|\tilde f(\delta(x))|&=&|<H\circ
\pi(x)\xi_0,\xi_0>-<H\xi_0,\pi(x^*)\xi_0>| \\
&\leq&\|H\xi_0\|\left(<\pi(x)\xi_0,\pi(x)\xi_0>^{1/2}+<\pi(x^*)\xi_0,\pi(x^*)\xi_0>^{1/2}\right),
\end{eqnarray*}
so that inequality (\ref{45}) follows with $C= \|H\xi_0\|$.

\vspace{3mm}

Let us now prove that (ii) implies (iii). For that we define a
sesquilinear form $\varphi$ in the following way: let $a,b$ be in
$\A$ and let $\{x_\alpha\},\{y_\beta\}$ be two nets in $\Ao$,
$\tau$-converging respectively to $a$ and $b$. We put

\begin{equation}
\varphi(a,b)=\lim_{\alpha,\beta}f(y_\beta^*x_\alpha). \label{411}
\end{equation}

It is readily checked that $\varphi$ is well-defined. The proofs
of (\ref{46}), (\ref{47}) and (\ref{48}) are easy consequences of
definition (\ref{411}) together with the properties of $f$.

To conclude the proof, we still have to check that (iii) implies
(i).

Given $\varphi$ as in (iii) above, we consider the
GNS-construction generated by $\varphi$.

Let ${\cal N}_\varphi=\{a\in\A; \varphi(a,a)=0\}$, then $\A/{\cal
N}_\varphi=\{\lambda_\varphi(a)=a+{\cal N}_\varphi, a\in\A\}$ is a
pre-Hilbert space with inner product
$<\lambda_\varphi(a),\lambda_\varphi(b)>=\varphi(a,b)$,
$a,b\in\lambda_\varphi(\A)$. We call $\Hil_\varphi$ the completion
of $\lambda_\varphi(\A)$ in the norm $\|.\|_\varphi$ given by this
inner product. It is easy to check that $\lambda_\varphi(\Ao)$ is
$\|.\|_\varphi$-dense in $\Hil_\varphi$. In fact, due to the
definition of locally convex quasi *-algebra, given $a\in\A$,
there exists a net ${x_\alpha}\subset \Ao$ such that
$x_\alpha\stackrel{\tau}{\rightarrow} a$. Therefore we have, using
the continuity of $\varphi$
$$
\|\lambda_\varphi(a)-\lambda_\varphi(x_\alpha)\|^2_\varphi=\|\lambda_\varphi(a-x_\alpha)\|^2_\varphi=\varphi(a-x_\alpha,a-x_\alpha)\leq
p(a-x_\alpha)^2\rightarrow 0.
$$
We can now define a  *-representation $\pi_\varphi$ with
ultra-cyclic vector $\lambda_\varphi(\1)$ as follows:
\begin{equation}
\pi_\varphi(a)\lambda_\varphi(x)=\lambda_\varphi(ax), \quad
a\in\A, x\in\Ao. \label{412}
\end{equation}
In particular, the fact that $\lambda_\varphi(\1)$ is ultra-cyclic
follows from the fact that
$\pi_\varphi(\Ao)\lambda_\varphi(\1)=\lambda_\varphi(\Ao)$  is
dense in $\Hil_\varphi$. Moreover the representation $\pi_\varphi$
is also $(\tau-\tau_s)$-continuous; in fact, taking $a\in\A$ and
$x\in\Ao$, we have:
$$
\|\pi_\varphi(a)\lambda_\varphi(x)\|^2_\varphi=\|\lambda_\varphi(ax)\|^2_\varphi=\varphi(ax,ax)\leq
(p(ax))^2\leq \gamma_x(p'(a))^2.
$$
The last inequality follows from the continuity of the
multiplication. This inequality shows that whenever
$\tau-\lim_{\alpha}x_\alpha=a$, then
$\tau_s-\lim_{\alpha}\pi_\varphi(x_\alpha)=\pi_\varphi(a)$.

This construction produces a *-representation $\pi_\varphi$ with
all the properties required to $\pi$ in (i). As a consequence, we
can define a *-derivation $\delta_{\pi_\varphi}$ induced by
$\pi_\varphi$ as in (\ref{41}):
$\delta_{\pi_\varphi}(\pi_\varphi(x))=\pi_\varphi(\delta(x))$, for
$x\in\Ao$. The proof of the spatiality of $\delta_{\pi_\varphi}$
generalizes the proof of the analogous statement for C*-algebras
(see, e.g. \cite{brarob}).

Let $\overline{\Hil_\varphi}$ be the conjugate space of
$\Hil_\varphi$, with inner product
$$
<{\lambda_\varphi(x)},{\lambda_\varphi(y)}>_{\overline{\Hil_\varphi}}=<\lambda_\varphi(y),\lambda_\varphi(x)>_{\Hil_\varphi}.
$$
>From now on we will indicate with the same symbol $<.,.>$ all the
inner products, whenever no possibility of confusion arises.

Let ${\cal M}_\varphi$ be the subspace of
$\Hil_\varphi\oplus\overline{\Hil_\varphi}$ spanned by the vectors
$\{\lambda_\varphi(x),\lambda_\varphi(x^*)\}$, $x\in\Ao$. We
define a linear functional $F_\varphi$ on ${\cal M}_\varphi$ by
\begin{equation}
F_\varphi(\{\lambda_\varphi(x),\lambda_\varphi(x^*)\})=i\varphi(\delta(x),\1),
\quad x\in\Ao. \label{413}
\end{equation}
Inequality (\ref{48}), together with the equality
$\|\{\lambda_\varphi(x),\lambda_\varphi(x^*)\}\|^2=\varphi(x,x)+\varphi(x^*,x^*)$,
shows that $f_\varphi$ is indeed continuous, so that by Riesz's
Lemma, there exists a vector $\{\xi_1,\xi_2\}\in
\Hil_\varphi\oplus\overline{\Hil_\varphi}$ such that
$$
F_\varphi(\{\lambda_\varphi(x),\lambda_\varphi(x^*)\})=<\{\lambda_\varphi(x),\lambda_\varphi(x^*)\},\{\xi_1,\xi_2\}>
=<\lambda_\varphi(x),\xi_1>+<\xi_2,\lambda_\varphi(x^*)>.
$$

Using the invariance of $\varphi$ we also deduce that

 $$
F_\varphi(\{\lambda_\varphi(x),\lambda_\varphi(x^*)\})=i\varphi(\delta(x),\1)=-i\overline{\varphi(\delta(x^*),\1)},
$$
which, together with the previous result, gives
\begin{equation}
\frac{1}{i}\varphi(\delta(x),\1)=<\lambda_\varphi(x),\eta>-<\eta,\lambda_\varphi(x^*)>,
\quad x\in\Ao, \label{414}
\end{equation}
where we have introduced the vector $\eta$ as
\begin{equation}
\eta=\frac{\xi_2-\xi_1}{2}. \label{415}
\end{equation}
Now we define the operator $H$ in the following way:
\begin{equation}
H\lambda_\varphi(x)=\frac{1}{i}\lambda_\varphi(\delta(x))+\hat\pi_\varphi(x)\eta,
\quad x\in\Ao, \label{416}
\end{equation}
where $\hat\pi_\varphi$ indicates the extension of $\pi_\varphi$,
defined in the usual way, which we need to introduce since $\eta$
belongs to $\Hil_\varphi$ and not to $\D_{\pi_\varphi}$, in
general.

First of all, we notice that from (\ref{416})
$H\lambda_\varphi(\1)=\eta\in\Hil_\varphi$, as stated in (i).
Moreover, $H$ is also well-defined and symmetric since for all
$x,y\in\Ao$
\begin{eqnarray*}
\lefteqn{<H\pi_\varphi(x)\lambda_\varphi(\1),\pi_\varphi(y)\lambda_\varphi(\1)>-<\pi_\varphi(x)\lambda_\varphi(\1),H\pi_\varphi(y)\lambda_\varphi(\1)>}
\\
&=&<H\lambda_\varphi(x),\lambda_\varphi(y)>-<\lambda_\varphi(x),H\lambda_\varphi(y)> \\
&=&<\left(\frac{1}{i}\lambda_\varphi(\delta(x))+\hat\pi_\varphi(x)\eta\right),\lambda_\varphi(y)>-<\lambda_\varphi(x),\left(\frac{1}{i}\lambda_\varphi(\delta(y))+\hat\pi_\varphi(y)\eta\right)>\\
&=&\frac{1}{i}\left(\varphi(\delta(x),y)+\varphi(x,\delta(y))\right)+<\hat\pi_\varphi(x)\eta,\lambda_\varphi(y)>-<\lambda_\varphi(x),\hat\pi_\varphi(y)\eta>\\
&=&\frac{1}{i}\varphi(\delta(y^*x),\1)+<\eta,\lambda_\varphi(x^*y)>-<\lambda_\varphi(y^*x),\eta>=0.
\end{eqnarray*}
This last equality follows from equation (\ref{414}). We finally
have to prove that $H$ implements the derivation
$\delta_{\pi_\varphi}$. For this, let $x,y,z\in\Ao$. Then we have
\begin{eqnarray*}
\lefteqn{i(<H\circ\pi_\varphi(x)\lambda_\varphi(y),\lambda_\varphi(z)>-<\pi_\varphi(x)\circ
H\lambda_\varphi(y),\lambda_\varphi(y)>)} \\
&=&i(<H\lambda_\varphi(xy),\lambda_\varphi(z)>-<
H\lambda_\varphi(y),\lambda_\varphi(x^*y)>) \\
&=&i\left(<\frac{1}{i}\lambda_\varphi(\delta(xy))+\hat\pi_\varphi(xy)\eta,\lambda_\varphi(z)>-<
\frac{1}{i}\lambda_\varphi(\delta(y))+\hat\pi_\varphi(y)\eta,\lambda_\varphi(x^*z)>\right)\\
&=&\varphi(\delta(x)y,z)=<\pi_\varphi(\delta(x))\lambda_\varphi(\delta(y)),\lambda_\varphi(\delta(z))>.
\end{eqnarray*}
Again, we made use of equation (\ref{414}).
\end{proof}

\vspace{3mm} {\bf Remark:--} If we add to a spatial *-derivation
$\delta_0$ a {\em perturbation} $\delta_p$ such that
$\delta=\delta_0+\delta_p$ is again a *-derivation, it is
reasonable to consider the question as to whether $\delta$ is
still spatial. The answer is positive under very general (and
natural) assumptions: since $\delta_0$ is spatial, the above
Proposition states that there exists a positive linear functional
$f$ on $\Ao$ whose extension $\tilde f$ satisfies, among the
others, inequality (\ref{45}): $|\tilde f(\delta_0(x))|\leq
C(\sqrt{f(x^*x)}+\sqrt{f(xx^*)})$, for all $x\in\Ao$. If we
require that $\delta_p$ satisfies the inequality $|\tilde
f(\delta_p(x))|\leq |\tilde f(\delta_0(x))|$, for all $x\in\Ao$,
which is exactly what we expect since $\delta_p$ is {\em small
compared to } $\delta_0$, we first deduce that $\delta_p$ is
spatial and, since, for all $x\in\Ao$, $|\tilde f(\delta(x))|\leq
2C(\sqrt{f(x^*x)}+\sqrt{f(xx^*)})$, using the same Proposition we
deduce that $\delta$ is spatial too. If  $H, H_0$ and $H_p$ denote
the operators that implement, respectively,  $\delta, \delta_0$
and $\delta_p$, we also get the equality
$i[H,A]\psi=i[H_0+H_p,A]\psi$, for all $A\in\LL^\dagger(\D_\pi)$
and $\psi\in\D_\pi$.

\vspace{3mm} The problem of the spatiality of a derivation is
particularly interesting when dealing with quantum systems with
infinite degrees of freedom. The reason is that for these systems
we need to introduce a regularizing cut-off in their descriptions
and remove this cut-off only at the very end. Specifically,
something like this can happen: the physical system ${\cal S}$ is
associated to, say, the whole space ${\mathbb R}^3$. In order to
describe the dynamics of ${\cal S}$ the canonical approach (see
\cite{brarob} and references therein)  consists in considering a
subspace $V\subset {\mathbb R}^3$, the physical system ${\cal
S}_V$ which naturally lives in this region, and to write down the
so-called local hamiltonian $H_V$ for ${\cal S}_V$. This
hamiltonian is a self-adjoint bounded operator which implements
the infinitesimal dynamics $\delta_V$ of ${\cal S}_V$. To obtain
information about the dynamics for ${\cal S}$ we need to compute a
limit (in $V$) to {\em remove the cutoff}. This can be a problem
already at this infinitesimal level (see also \cite{bagtra} and
references therein) and becomes harder and harder, in general,
when the interest is moved to the finite form of the algebraic
dynamics, that is, when we try to integrate the derivation. Among
the other things, for instance, it may happen that the net $H_V$
or the related net $\delta_V$ (or both), does not converge in any
reasonable topology, or that $\delta_V$ is not spatial. Another
possibility that may occur is the following: $H_V$ converges (in
some topology) to a certain operator $H$, $\delta_V$ converges (in
some other topology) to a certain *-derivation $\delta$, but
$\delta$ is not spatial or, even if it is, $H$ is not the operator
which implements $\delta$.

However, under some reasonable conditions, all these possibilities
can be controlled. The situation is governed by the next
Proposition, which is based on the assumption that there exists a
$(\tau-\tau_s)$-continuous *-representation $\pi$ in the Hilbert
space $\Hil_\pi$, which is ultra-cyclic  with ultra-cyclic vector
$\xi_0$, and a family of *-derivations (in the sense of Definition
\ref{Definition 3.1}) $\{\delta_n: \, n\in \mathbb N\}$ of the
 *-algebra $\Ao$ with identity. We define
 a related family of *-derivations $\delta_\pi^{(n)}$
induced by $\pi$ defined on $\pi(\Ao)$ and with values in
$\pi(\A)$:

\begin{equation}
 \delta_\pi^{(n)}(\pi(x))=\pi(\delta_n(x)), \quad x\in\Ao.
\label{417}
\end{equation}

\begin{prop}

Suppose that:

\begin{itemize}
\item[(i)] $\{\delta_n(x)\}$ is $\tau$-Cauchy for all $x\in\Ao$;
\item[(ii)] For each $n\in {\mathbb N}$,  $\delta_\pi^{(n)}$ is spatial, that is, there exists an operator $H_n$ such that $$H_n=H_n^\dagger\in \LL(\D_\pi,\D_\pi'),$$
$H_n\xi_0\in\Hil_\pi$ and
$\delta_\pi^{(n)}(\pi(x))=i\{H_n\circ\pi(x)-\pi(x)\circ H_n\},
\forall x\in\Ao$;
\item[(iii)] $$\sup_n\|H_n\xi_0\|=:L<\infty.$$
\end{itemize}
Then:
\begin{itemize}
\item[(a)] $\exists \, \delta(x)=\tau-\lim\delta_n(x)$, for all $x\in\Ao$, which is a *-derivation of $\Ao$;
\item[(b)] $\delta_\pi$, the *-derivation induced by $\pi$, is well-defined and spatial;
\item[(c)] if $H$ is the self-adjoint operator which implements $\delta_\pi$, if $<(H_n-H)\xi_0,\xi>\rightarrow 0$ for all $\xi\in D_\pi$ then $H_n$ converges weakly to $H$.
\end{itemize}
\label{theorem42}
\end{prop}
\begin{proof} (a) This first statement is trivial.

\vspace{3mm} (b) For $a,b\in\A$ we put
$\varphi(a,b)=<\pi(a)\xi_0,\pi(b)\xi_0>$. Then $\varphi$ is an
invariant positive sesquilinear form on $\A\times\A$, since:
$$
\varphi(ax,y)=<\pi(ax)\xi_0,\pi(y)\xi_0>=<\pi(a)\pi(x)\xi_0,\pi(y)\xi_0>=
<\pi(x)\xi_0,\pi(a^*)\pi(y)\xi_0>=\varphi(x,a^*y),
$$
for all $a\in\A$ and $x,y\in\Ao$. $\varphi$ is $\tau$-continuous:
if $a,b\in\A$
$$
|\varphi(a,b)|=|<\pi(a)\xi_0,\pi(b)\xi_0>|\leq
\|\pi(a)\xi_0\|\|\pi(b)\xi_0\|\leq p_\alpha(a) p_\alpha(b),
$$
for some continuous seminorm $p_\alpha$ on $\A$, because of the
$(\tau-\tau_s)$-continuity of $\pi$.

>From this inequality we deduce that, for $x\in\Ao$,
\begin{eqnarray*}
|\varphi(\delta(x),\1)|&=&\lim_n|\varphi(\delta_n(x),\1)|=\lim_n|<H_n\circ
\pi(x)\xi_0,\xi_0>-<\pi(x)\circ H_n\xi_0,\xi_0>| \\
&=&\limsup_n|<H_n\circ \pi(x)\xi_0,\xi_0>-<\pi(x)\circ
H_n\xi_0,\xi_0>| \\ &\leq& \limsup_n
\|H_n\xi_0\|(\|\pi(x)\xi_0\|+\|\pi(x^*)\xi_0\|) \\
&\leq& L (\sqrt{\varphi(x,x)}+\sqrt{\varphi(x^*,x^*)}).
\end{eqnarray*}

This sesquilinear form $\varphi$ satisfies all the conditions
required in  (iii) of Theorem \ref{theorem41}. Then, following the
same steps as in the proof of Theorem \ref{theorem41}, (iii)
$\Rightarrow$ (i), we construct the GNS-representation
$\pi_\varphi$ associated to $\varphi$. We call $\Hil_\varphi,
\xi_\varphi$ and $H_\varphi$  respectively the Hilbert space, the
ultra-cyclic vector and the symmetric operator implementing the
derivation associated to  $\pi_\varphi$. Among others, the
following equality must be satisfied:
\begin{equation}
\varphi(a,b)=<\pi(a)\xi_0,\pi(b)\xi_0>=<\pi_\varphi(a)\xi_\varphi,\pi_\varphi(b)\xi_\varphi>,
\quad \forall a,b\in \A, \label{418}
\end{equation}
which implies that $\pi_\varphi$ and $\pi$ are unitarily
equivalent, that is, there exists a unitary operator
$U:\Hil_\pi\rightarrow \Hil_\varphi$ such that
$U\xi_0=\xi_\varphi$, $U\pi(a)U^{-1}=\pi_\varphi(a)$, $\forall
a\in\A$, and $U$ is continuous from $D_\pi[t_\pi]$ into
$D_\varphi[t_\varphi]$. We prove here only this last property. Let
$x,y\in\Ao$; we have
$$
\|\pi_\varphi(y)U\pi(x)\xi_0\|_\varphi=\|U\pi(y)\pi(x)\xi_0\|_\varphi=\|\pi(y)\pi(x)\xi_0\|,
$$
which implies that $U^*$ can be extended to an operator
$U^\dagger:\D_\varphi'\rightarrow\D_\pi'$. We have now
$$
\delta_{\pi_\varphi}(\pi_\varphi(x))=\pi_\varphi(\delta(x))=U
\pi(\delta(x))U^{-1}=U \delta_{\pi}(\pi(x))U^{-1},
$$
which implies that
$\delta_{\pi}(\pi(x))=U^{-1}\delta_{\pi_\varphi}(\pi_\varphi(x))U$.
Since $\delta_{\pi_\varphi}$ is well-defined, this equality
implies that also $\delta_\pi$ is well-defined. Indeed we have:
$$
\pi(x)=0 \Rightarrow \pi_\varphi(x)=0 \Rightarrow
\delta_{\pi_\varphi}(\pi_\varphi(x))=0 \Rightarrow
\delta_{\pi}(\pi(x))=0.
$$
Now we define $H=U^{-1}H_\varphi U|_{\D_\pi}$. Then
\begin{eqnarray*}
\delta_{\pi}(\pi(x))&=&U^{-1}\delta_{\pi_\varphi}(\pi_\varphi(x))U
=iU^{-1}\left(H_\varphi\circ\pi_\varphi(x)-\pi_\varphi(x) \circ
H_\varphi\right)U \\&=&i\left(U^{-1}H_\varphi U\circ
U^{-1}\pi_\varphi(x)U-U^{-1}\pi_\varphi(x)U\circ U^{-1}H_\varphi
U\right)\\ &=& i\left(H\circ\pi(x)-\pi(x)\circ H\right),
\end{eqnarray*}
which allows us to conclude.

\vspace{3mm}

(c) For $x,y,z\in\Ao$ we have, using the definition of $\varphi$
and its $\tau$-continuity,
\begin{eqnarray*}
\varphi(\delta_n(x)y,z)&=&<\delta_\pi^{(n)}(\pi(x))\pi(y)\xi_0,\pi(z)\xi_0>
\\
&=& i\left(<(H_n\circ \pi(x))\pi(y)\xi_0,\pi(z)\xi_0>-<(
\pi(x)\circ H_n)\pi(y)\xi_0,\pi(z)\xi_0>\right)\\ & \rightarrow
&\varphi(\delta(x)y,z).
\end{eqnarray*}

Since $<( \pi(x)\circ H_n)\pi(y)\xi_0,\pi(z)\xi_0>=<
H_n\pi(y)\xi_0,\pi(x^*z)\xi_0>$, we deduce that, taking $y=\1$,
$<( \pi(x)\circ H_n)\xi_0,\pi(z)\xi_0>= <
H_n\xi_0,\pi(x^*z)\xi_0>\rightarrow < H\xi_0,\pi(x^*z)\xi_0>$,
because of the assumption on $H_n$. Then, by difference with
$$\varphi(\delta(x),z)=i{<(H\circ\pi(x))\xi_0,\pi(z)\xi_0>-<(\pi(x)\circ
H)\xi_0,\pi(z)\xi_0>},$$ we get that
$<(H_n\pi(x))\xi_0,\pi(z)\xi_0>\rightarrow
<(H\pi(x))\xi_0,\pi(z)\xi_0>$, for all $x,z\in\Ao$. Then $H_n$
converges to $H$ weakly.

\end{proof}

\vspace{5mm} {\bf Example 1:} A radiation model

\vspace{2mm} In this example the representation $\pi$ is just the
identity map. Let us consider a model of $n$ free bosons,
\cite{bagjmp}, whose dynamics is given by the hamiltonian,
$H=\sum_{i=1}^n a_i^\dagger a_i$. Here $a_i$ and $a_i^\dagger$ are
respectively  the annihilation and creation operators for the
$i$-th mode. They satisfy the following CCR
\begin{equation}
[a_i,a_j^\dagger]=\1\delta_{i,j}. \label{419}
\end{equation}
Let $Q_L$ be the projection operator on the subspace of $\Hil$
with at most $L$ bosons. This operator can be written considering
the spectral decomposition of $H_{(i)}=a_i^\dagger
a_i=\sum_{l=0}^\infty l E_l^{(i)}$. We have
$Q_L=\sum_{i=1}^n\sum_{l=0}^L E_l^{(i)}$. Let us now define a
bounded operator $H_L$ in $\Hil$ by $H_L=Q_LHQ_L$. It is easy to
check that, for any vector $\Phi_M$ with $M$ bosons (i.e., an
eigenstate of the number operator $N=H=\sum_{i=1}^n a_i^\dagger
a_i$ with eigenvalue $M$), the condition
$\sup_L\|H_L\Phi_M\|<\infty$ is satisfied. In particular, for
instance, $\sup_L\|H_L\Phi_0\|=0$. It may be worth remarking that
all the vectors $\Phi_M$ are cyclic. Denoting with $\delta_L$ the
derivation implemented by $H_L$ and $\delta$ the one implemented
by $H$, it is clear that all the assumptions of the previous
Proposition are satisfied, so that, in particular, the weak
convergence of $H_L$ to $H$ follows. This is not surprising since
it is known that $H_L$ converges to $H$ strongly on a dense
domain, \cite{bagjmp}. .

\vspace{5mm} {\bf Example 2:} A mean-field spin model

\vspace{2mm} The situation described here is quite different from
the one in the previous example. First of all, \cite{bagarello,
bagatrap}, there exists no hamiltonian for the whole physical
system but only for a finite volume subsystem:
$H_V=\frac{1}{|V|}\sum_{i,j\in V}\sigma_3^i\sigma_3^j$, where $i$
and $j$ are the indices of the lattice site, $\sigma_3^i$ is the
third component of the Pauli matrices, $V$ is the volume cut-off
and $|V|$ is the number of the lattice sites in $V$. It is
convenient to introduce the {\em mean magnetization} operator
$\sigma_3^V=\frac{1}{|V|}\sum_{i\in V}\sigma_3^i$. Let us indicate
with $\uparrow_i$ and $\downarrow_i$ the eigenstates of
$\sigma_3^i$ with eigenvalues $+1$ and $-1$, respectively. We
define $\Phi_\uparrow=\otimes_{i\in V}\uparrow_i$. It is clear
that $\sigma_3^V\Phi_\uparrow=\Phi_\uparrow$, which implies that
$H_V\Phi_\uparrow=|V|\Phi_\uparrow$, which in turns implies that
$\sup_V\|H_V\Phi_\uparrow\|=\infty$. This means that the cyclic
vector $\Phi_\uparrow$ does not satisfy the main assumption of
Proposition \ref{theorem42}, and for this reason nothing can be
said about the convergence of $H_V$. However, it is possible to
consider a different cyclic vector
$$
\Phi_0=....\otimes \uparrow_{j-1} \otimes \downarrow_j\otimes
\uparrow_{j+1}\otimes \downarrow_{j+2}\otimes...,
$$
which is again an eigenstate of $\sigma_3^V$. Its eigenvalue
depends on the volume $V$. However, it is clear that
$\|\sigma_3^V\Phi_0\|=\frac{1}{|V|}\|\Phi_0\|\epsilon_V$, where
$\epsilon_V$ can take only the values $0, 1$. Analogously we have
$\|H_V\Phi_0\|=\frac{1}{|V|}\|\Phi_0\|\epsilon_V^2\rightarrow 0$.
This means that this vector satisfies the assumptions of
Proposition \ref{theorem42}, so that the derivation
$\delta_V(.)=i[H_V,.]$ converges to a derivation $\delta$ which is
spatial and implemented by $H$, and that $H_V$ is weakly
convergent to $H$.

As we see, contrary to the previous example, the choice of the
cyclic vector which we take as our starting point, is very
important in order to be able to prove the existence of $\delta$,
its spatiality and convergence of $H_V$ to a limit operator. It is
also worth remarking that the same conclusions could also be found
replacing $\Phi_0$ with any vector which can be obtained as a
local perturbation of $\Phi_0$ itself.

\vspace{3mm} {\bf Remark:--} All the results we have proved above
can be specialized to a CQ*-algebras, which can be considered as
particular example of locally convex quasi *-algebras. The main
difference in this case concerns statement $(c)$ of Proposition
\ref{theorem42}: the weak convergence of $H_n$ to $H$, in this
case, is replaced by a strong convergence. More in details,
referring to the Example of Section 2 and calling
$\Omega\in\Hil_{+1}$ a cyclic vector, we can prove that, if
$\|(H_n-H)\Omega\|_{-1}\rightarrow 0$, then
$\|(H_n-H)A\Omega\|_{-1}\rightarrow 0$ for all $A\in
B(\Hil_{+1})$.

\vspace{5mm} The following result gives an interplay between the
results of this and of the previous sections. In particular, we
consider now the possibility of extending the domain of definition
of the derivation $\delta$ (as we did in Section 3) defined as a
limit of a net of derivations $\delta_n$ (as we have done in this
section). For this we first need the following definition:

\begin{defn}
Let $(\A[\tau],\Ao)$ be a locally convex quasi *-algebra. A
sequence $\{\delta_n\}$ of *-derivations is called uniformly
$\tau$-continuous if, for any continuous seminorm $p$ on $\A$,
there exists a continuous seminorm $q$ on $\A$ such that

\begin{equation}
p(\delta_n(x))\leq q(x), \quad\forall x\in\Ao, \forall n\in
\mathbb{N}. \label{420}
\end{equation}
\label{Definition 4.1}
\end{defn}

We can now prove the following

\begin{prop}

Let $\delta$ be the $\tau$-limit of a  uniformly $\tau$-continuous
sequence  $\{\delta_n\}$ of *-derivations such that the set

\begin{equation}
\D(\delta)=\{x\in\Ao: \exists \,\tau-\lim_n\delta_n(x)\}
\label{421}
\end{equation}
is $\tau$-dense in $\Ao$. Then, $\delta$ is a *-derivation and,
denoting with $\tilde\delta_n$ the continuous extension of
$\delta_n$ to $\A$, we have: $\{x\in\A: \exists
\tau-\lim_n\tilde\delta_n(x)\}=\A$.

\end{prop}

\begin{proof} The proof that $\delta$ is a *-derivation is trivial.

Let $a$ be a generic element in $\A$. Since, by assumption,
$\D(\delta)$ is $\tau$-dense in $\Ao$, and therefore in $\A$,
there exists a net $\{x_\alpha\}\subset D(\delta)$
$\tau$-converging to $a$. This means that for any continuous
seminorms $p$ and for any $\epsilon>0$ there exists
$\alpha_{p,\epsilon}$ such that $p(a-x_\alpha)<\epsilon$ for all
$\alpha>\alpha_{p,\epsilon}$.

Take an arbitrary continuous seminorm $p$ on $\A$. Let $q$ be the
continuous seminorm on $\A$ satisfying (\ref{420}). Then,

\begin{eqnarray*}
p(\tilde\delta_n(a)-\tilde\delta_m(a))&\leq&
p(\tilde\delta_n(a-x_\alpha))+p((\tilde\delta_n-\tilde\delta_m)(x_\alpha))+
p(\tilde\delta_m(a-x_\alpha))\\
&\leq&
2q(a-x_\alpha))+p((\tilde\delta_n-\tilde\delta_m)(x_\alpha))\\
&\leq& 2\epsilon +p((\tilde\delta_n-\tilde\delta_m)(x_\alpha))\leq
\epsilon',
\end{eqnarray*}
for all fixed $\alpha>\alpha_{q,\epsilon}$ and $n,m$ large enough.
This completes the proof.

\end{proof}

\vspace{3mm}

All the results obtained in this section rely on the fact that
there exists one  underlying Hilbert space related to the
representation, in the case of locally convex quasi *-algebras, or
to triplets of Hilbert spaces for CQ*-algebras. However, it is
known that in some physically relevant situation like in quantum
field theory, the relevant operators are the quantum fields and
these operators belong to $\LL(\D,\D')$ for suitable $\D$, instead
of being in some $\LL^\dagger(\D,\Hil)$. This motivates our
interest for the next result, which extends in a non trivial way
Proposition \ref{theorem42}. Before stating the Proposition, we
need to introduce some definitions.

Let $(\A,\Ao)$ be a quasi *-algebra and $\pi_0$ a *-representation
of $\Ao$ on the domain $\D_{\pi_0}\subset \Hil_{\pi_0}$. This
means that $\pi_0$ maps $\Ao$ into $\LL^\dagger(\D_{\pi_0})$ and
that $\pi_0$ is a *-homomorphism of *-algebras. As usual, we endow
$\D_{\pi_0}$ with the topology $t_\dagger$, the graph topology
generated by $\LL^\dagger(\D_{\pi_0})$: in this way we get the
rigged Hilbert space $\D_{\pi_0}\subset \Hil_{\pi_0}\subset
\D_{\pi_0}'$, where $\D_{\pi_0}'$ is the dual of
$\D_{\pi_0}[t_\dagger]$. On $\D_{\pi_0}'$ we consider the strong
dual topology $t_\dagger'$ defined by the seminorms
\begin{equation}
\|F\|_{\cal M}=\sup_{\xi\in{\cal M}}|<F,\xi>|, \quad\quad {\cal M}
\mbox{ bounded in } \D_{\pi_0}[t_\dagger]. \label{422}
\end{equation}
In $\LL(\D_{\pi_0},\D_{\pi_0}')$ we consider  the {\em
quasi-strong topology} $\tau_{qs}$ defined by the seminorms
$$
\LL(\D_{\pi_0},\D_{\pi_0}')\ni X\rightarrow \|X\xi\|_{\cal M},
\quad\quad \xi\in\D_{\pi_0}, {\cal M} \mbox{ bounded in }
\D_{\pi_0}[t_\dagger];
$$
and the {\em uniform topology} $\tau_{\D}$, defined by the
seminorms
$$
\LL(\D_{\pi_0},\D_{\pi_0}')\ni X\rightarrow \|X\|_{\cal
M}=\sup_{\xi,\eta\in{\cal M}}|<X\xi,\eta>|, \quad  {\cal M} \mbox{
bounded in } \D_{\pi_0}[t_\dagger].
$$
\begin{defn}

Let $(\A,\Ao)$ and $\pi_0$ be as above. A linear map
$\pi:\A\to\LL(\D_{\pi},\D_{\pi}')$ is called a {\em
qu*-representation} of $\A$ associated with $\pi_0$  if $\pi$
extends $\pi_0$ and

$\pi(a^*)=\pi(a)^\dagger \quad \forall a\in \A$;

$\pi(ax)=\pi(a)\pi_0(x) \quad \forall a\in \A, x\in\Ao$.
\label{Definition 4.2}
\end{defn}
\begin{thm}

Let $(\A,\Ao)$ be a locally convex quasi *-algebra with identity
and with topology $\tau$ and $\delta$ be a *-derivation of $\Ao$.

Then the following statements are equivalent:

(i) There exists a $(\tau-\tau_{qs})$-continuous, ultra-cyclic
qu*-representation $\pi$ of $(\A,\Ao)$, with ultra-cyclic vector
$\xi_0$ such that the *-derivation $\delta_\pi$ induced by  $\pi$
is spatial, i.e. there exists $H=H^\dagger\in
\LL(\D_{\pi},\D_{\pi}')$ such that
\begin{equation}
\delta_\pi(\pi(x))=i\{H\circ\pi(x)-\pi(x)\circ H\},\quad \forall
x\in\Ao. \label{423}
\end{equation}

\vspace{2mm} (ii) There exists a positive linear functional $f$ on
$\Ao$ and a sesquilinear positive form $\Omega$ on $\Ao\times\Ao$
such that:

\begin{itemize}
\item[(a)] for some continuous seminorm p on $\A[\tau]$,
\begin{equation}
f(x^*x)\leq p(x)^2, \quad \forall x\in \Ao, \label{424}
\end{equation}
\item[(b)]
Let $\tilde f$ be the continuous extension of $f$ to  $\A$, then
the following inequalities hold:
\begin{equation}
|\tilde f(y^*x)|\leq p(x) \Omega(y,y)^{1/2}, \quad \forall x,y\in
\Ao, \label{425}
\end{equation}
for some continuous seminorm $p$;
\item[(c)]
\begin{equation}
|\tilde f(y^*ax)|\leq \gamma_a\Omega(x,x)^{1/2}\Omega(y,y)^{1/2},
\quad \forall x,y\in \Ao, a\in\A \label{426}
\end{equation}
for some positive constant $\gamma_a$;
\item[(d)]
\begin{equation}
|\tilde f(\delta(x))|\leq
C(\Omega(x,x)^{1/2}+\Omega(x^*,x^*)^{1/2}), \quad \forall x\in
\Ao, \label{427}
\end{equation}
for some positive constant $C$.
\item[(e)]
For any ultra-cyclic *-representation $\Theta$ of $\Ao$, with
ultra-cyclic vector $\xi_\theta$, satisfying
$$f(x)=<\Theta(x)\xi_\theta,\xi_\theta>,$$ for all $x\in \Ao$, the
sesquilinear form on $\D_\theta\times\D_\theta$,
$\D_\theta=\Theta(\Ao)\xi_\theta$, defined by
$$\varphi_\theta(\Theta(x)\xi_\theta,
\Theta(y)\xi_\theta)=\Omega(x,y)$$ is jointly continuous on
$\D_\theta[t_\dagger]$.
\end{itemize}
\label{theorem43}

\end{thm}

\begin{proof}
Let us prove that (i) implies (ii). For this, let $\pi$ be a
$(\tau-\tau_{qs})$-continuous, ultra-cyclic qu*-representation  of
$\A$ associated with $\pi_0$, with ultra-cyclic vector $\xi_0$:
$\pi_0(\Ao)\xi_0=\D_\pi$. For all $x\in\Ao$ we define
$f(x)=<\pi_0(x)\xi_0,\xi_0>$. Then, since $\pi$ coincides with
$\pi_0$ on $\Ao$ and since $\pi$ is $(\tau-\tau_{qs})$-continuous,
we have

$$
f(x^*x)=<\pi_0(x^*x)\xi_0,\xi_0>=<\pi(x^*x)\xi_0,\xi_0>=\|\pi(x)\xi_0\|^2\leq
p(x)^2,
$$
for some continuous seminorm $p$ of $\A[\tau]$. In fact,
$\|\pi(x)\xi_0\|$ is one of the seminorms defining $\tau_{qs}$.
Calling $\tilde f$ the continuous extention of $f$ it is clear
that, for any $a\in \A$, we have $\tilde
f(a)=<\pi(a)\xi_0,\xi_0>$. Therefore, for $x,y\in\Ao$  and
$a\in\A$, we have
$$
\tilde
f(y^*ax)=<\pi(y^*ax)\xi_0,\xi_0>=<\pi(ax)\xi_0,\pi_0(y)\xi_0>=<\pi(a)\pi_0(x)\xi_0,\pi_0(y)\xi_0>,
$$
and, since by assumption $\pi(a)\pi_0(x)\xi_0$ is a continuous
functional on $\D[t_\dagger]$, there exists a positive constant
$\gamma$ and a continuous seminorm on $\D[t_\dagger]$ such that
$$
|\tilde f(y^*ax)|\leq \gamma \|T\pi_0(y)\xi_0\|,
$$
where $T\in\LL^\dagger (\D_\pi)$ labels the seminorm. The best
value of $\gamma$ can be found considering the following bounded
subset ${\cal M}$ of $\D_\pi[t_\dagger]$: ${\cal M}=\{\xi\in
\D_\pi: \|T\xi\|=1\}$. In this way we get
\begin{equation}
|\tilde f(y^*ax)|\leq \|\pi(a)\pi_0(x)\xi_0\|_{\cal M}
\|T\pi_0(y)\xi_0\|\leq p_x(a) \|T\pi_0(y)\xi_0\|. \label{428}
\end{equation}
The last inequality follows from the $(\tau-\tau_{qs})$-continuity
of $\pi$. Furthermore, since $\pi(a)$ belongs to
$\LL(\D_{\pi},\D_{\pi}')$, the following inequality also holds:
\begin{equation}
\|\pi(a)\pi_0(x)\xi_0\|_{\cal M} \leq \gamma_2 \|C\pi(x)\xi_0\|
\label{429}
\end{equation}
for a certain positive constant $\gamma_2$ and an operator
$C\in\LL^\dagger (\D_\pi)$.

Moreover, since $\tilde
f(\delta(x))=i\{<\pi(x)\xi_0,H\xi_0>-<H\xi_0,\pi(x^*)\xi_0>\}$,
and since, as a functional, $H\xi_0$ is continuous, there exists a
$B\in\LL^\dagger(\D_\pi)$ and a positive constant $\gamma_1$ such
that

\begin{equation}
|\tilde f(\delta(x))|\leq \gamma_1
(\|B\pi(x)\xi_0\|+\|B\pi(x^*)\xi_0\|). \label{430}
\end{equation}

The above inequalities refer to three elements of $\LL^\dagger
(\D_\pi)$, $B,C$ and $T$. It is always possible to find another
element $A\in \LL^\dagger (\D_\pi)$ such that

\begin{equation}
\|A\eta\|\geqq \|B\eta\|, \quad \|A\eta\|\geqq \|T\eta\|, \quad
\|A\eta\|\geqq \|C\eta\|, \quad \forall \eta\in\D_\pi. \label{431}
\end{equation}

Let us now define the positive sesquilinear form $\Omega$ on
$\Ao\times\Ao$ as

\begin{equation}
\Omega(x,y)=<A\pi(x)\xi_0,A\pi(y)\xi_0>, \quad x,y\in\Ao.
\label{432}
\end{equation}

Then, because of (\ref{431}), inequalities (\ref{424})-(\ref{427})
easily follow. As for the joint continuity of $\varphi_\theta$, we
start noticing that, since
$f(x)=<\pi_0(x)\xi_0,\xi_0>=<\Theta(x)\xi_\theta,\xi_\theta>$,
then $\Theta$ is unitarily equivalent to $\pi_0$, since they are
both unitarily equivalent to the GNS representation $\pi_f$
defined by $f$ on $\Ao$, because of the essential uniqueness of
the latter. Thus, there exists a unitary operator
$U:\Hil_\theta\rightarrow\Hil_{\pi_0}$, with $\xi_0=U\xi_\theta$
and such that $\Theta(x)=U^{-1}\pi_0(x)U$.

By the definition itself,
$$
\varphi_{\pi_0}(\pi_0(x)\xi_0,
\pi_0(y)\xi_0)=\Omega(x,y)=<A\pi_0(x)\xi_0, A\pi_0(y)\xi_0>
$$
then $\varphi_{\pi_0}$ is jointly continuous on
$\D_{\pi_0}[t_\dagger]$. Therefore
$$
\varphi_{\theta}(\Theta(x)\xi_\theta,
\Theta(y)\xi_\theta)=\Omega(x,y)=<A\pi_0(x)\xi_0,
A\pi_0(y)\xi_0>=<U^{-1}AU\Theta(x)\xi_\theta,
U^{-1}AU\Theta(y)\xi_\theta>
$$
and $\varphi_{\theta}$ is jointly continuous on
$\D_{\theta}[t_\dagger]$, too.

\vspace{3mm}

We prove now the converse implication,i.e.  (ii) implies (i).

We assume that there exist $f$ and $\Omega$ satisfying all the
properties we have required in (ii).  We define the following
vector space: ${\cal N}_f=\{a\in\A:\, \tilde f(a^*x)=0\,\,\forall
x\in\Ao\}$. It is clear that if $a\in {\cal N}_f$ and $y\in\Ao$,
then $ya\in{\cal N}_f$. We denote with $\lambda_f(a)$, for
$a\in\A$, the element of the vector space $\A/{\cal N}_f$
containing $a$. The subspace $\lambda_f(\Ao)=\{\lambda_f(x),
\,x\in\Ao\}$ is a pre-Hilbert space with inner product
$$
<\lambda_f(x),\lambda_f(y)>=f(y^*x), \quad x,y\in\Ao
$$
and the form $<\lambda_f(x),\lambda_f(a)>=\tilde f(a^*x)$,
$x\in\Ao$, $a\in\A$, puts $\A/{\cal N}_f$ and $\lambda_f(\Ao)$ in
separating duality. Now we can define a ultra-cyclic
*-representation $\pi_0$ of $\Ao$ in the following way: its domain
$\D_{\pi_0}$ coincides with $\lambda_f(\Ao)$, and
$\pi_0(x)\lambda_f(y)=\lambda_f(xy)$, for $x,y\in\Ao$. The vector
$\lambda_f(\1)$ is ultra-cyclic  and
$f(x)=<\pi_0(x)\lambda_f(\1),\lambda_f(\1)>$, for all $x\in\Ao$.
Therefore the sesquilinear form
$\varphi_{\pi_0}(\pi_0(x)\lambda_f(\1),\pi_0(y)\lambda_f(\1))=\Omega(x,y)$
is jointly continuous in $\D_{\pi_0}[t_\dagger]$.

We now claim that $\A/{\cal N}_f\subset \D_{\pi_0}'$, the dual
space of $\D_{\pi_0}[t_\dagger]$. This follows from the joint
continuity of $\varphi_{\pi_0}$, which gives the following
estimate
\begin{equation}
|\Omega(x,y)|\leq \gamma \|A'\pi_0(x)\lambda_f(\1)\|
\|A'\pi_0(y)\lambda_f(\1)\| \label{433}
\end{equation}
which holds for all $x,y\in\Ao$, for suitable $\gamma>0$ and
$A'\in \LL^\dagger(\D_{\pi_0})$. Using the extension of
(\ref{425}) to $\Ao\times\A$ and (\ref{433}) we find

$$
|<\lambda_f(x),\lambda_f(a)>|=|\tilde f(a^*x)|\leq p(a)
\Omega(x,x)^{1/2}\leq \gamma^{1/2}p(a)
\|A'\pi_0(x)\lambda_f(\1)\|,
$$
which implies that $\lambda_f(a)\in\D_{\pi_0}'$.

We can now extend $\pi_0$ to $\A$ in a natural way: for $a\in \A$
we put $\pi(a)\lambda_f(x)=\lambda_f(ax)$, for all $x\in\Ao$. For
each $a \in \A$, $\pi(a)$ is well-defined and maps
$\D_{\pi_0}[t_\dagger]$ into $\D_{\pi_0}'[t_\dagger']$
continuously. Moreover $\pi$ is $(\tau-\tau_{qs})$-continuous. The
induced derivation $\delta_\pi$ is well-defined, as is easily
checked, and its spatiality can be proven by repeating essentially
the same steps as in Proposition \ref{theorem41}.
\end{proof}

\vspace{4mm} {\bf Remark:--} In the so-called Wightman formulation
 of quantum field theory {see, e.g. \cite{haag}), the point-like
$A(x)$, $x \in {\mb R}^4$, can be a very singular mathematical
object such as a sesquilinear form depending on $x$ and defined on
$\D \times \D$, where $\D$ is a dense domain in Hilbert space
$\Hil$. The {\em smeared field} is an operator-valued distribution
$f \in {\mc S}({\mb R}^4)\to \LD$, ${\mc S}({\mb R}^4)$ being the
space of Schwartz test functions. If $f$ has support contained in
a bounded region ${\mc O}$ of ${\mb R}^4$, then $\overline{A(f)}$
is affiliated with the local von Neumann algebra $\A({\mc O})$ of
all observables in ${\mc O}$.

A reasonable approach \cite{fredhert, epitrafields} consists in
considering the point-like field $A(x)$, for each $x \in {\mb
R}^4$, as an element of $\LDD$, once a locally convex topology on
$\D$ has been defined. A crucial physical prescription is that the
field must be covariant under the action of a unitary
representation $U(g)$ of some transformation group (such as the
Poincar\'e or Lorentz group) and, as is known, the infinitesimal
generator $H$ of time translations gives the energy operator of
the system which defines in natural way a spatial *-derivation of
the quasi *-algebra $(\A, \Ao)$ of observables.

There could be however a different approach. This occurs when a
field $x \mapsto A(x)$ is defined on the basis of some heuristic
considerations. In order that $A(x)$ represent a reasonable
physical solution of the problem under consideration, covariance
under some Lie algebra of infinitesimal transformation must be
imposed. For the infinitesimal time translations this amounts to
find some *-derivation $\delta$ of the quasi *-algebra obtained by
taking the weak completion of the *-algebra $\Ao$ generated by the
local von Neumann algebras $\A({\mc O})$, with ${\mc O}$ a bounded
region of ${\mb R}^4$. But, of course, a number of problems arise.

The first one consists in finding an appropriate domain $\D$ for
the family of  operators $\{  A(f);\, f \in {\mc S}({\mb R}^4)\}$
and an appropriate topology on $\D$, in such a way that $A(x) \in
\LDD$ for every $x \in {\mb R}^4$. Once this is done, if the
identical representation has the properties required  in Theorem
\ref{theorem43}, then a symmetric operator $H$ implementing
$\delta$ can be found and one expects $H$ to be the energy
operator of the system. But, as is well-known, the problem of
integrating $\delta$ is far to be solved even in much more regular
situations than those considered here. We hope to discuss these
problems in a future paper.

\vspace{6mm} \noindent{\large \bf Acknowledgments} \vspace{5mm}

We would like to thank Prof. K. Schm{\"u}dgen for his useful
remarks.

We acknowledge the financial support  of the Universit\`a degli
Studi di Palermo (Ufficio Relazioni Internazionali) and of the
Italian Ministry of Scientific Research.

\end{document}